\PassOptionsToPackage{unicode}{hyperref}
\PassOptionsToPackage{hyphens}{url}
\PassOptionsToPackage{dvipsnames,svgnames,x11names}{xcolor}
\documentclass[
  12pt]{article}

\usepackage{amsmath,amssymb}
\usepackage{iftex}
\ifPDFTeX
  \usepackage[T1]{fontenc}
  \usepackage[utf8]{inputenc}
  \usepackage{textcomp} 
\else 
  \usepackage{unicode-math}
  \defaultfontfeatures{Scale=MatchLowercase}
  \defaultfontfeatures[\rmfamily]{Ligatures=TeX,Scale=1}
\fi
\usepackage{lmodern}
\ifPDFTeX\else  
\fi
\IfFileExists{upquote.sty}{\usepackage{upquote}}{}
\IfFileExists{microtype.sty}{
  \usepackage[]{microtype}
  \UseMicrotypeSet[protrusion]{basicmath} 
}{}
\makeatletter
\@ifundefined{KOMAClassName}{
  \IfFileExists{parskip.sty}{%
    \usepackage{parskip}
  }{
    \setlength{\parindent}{0pt}
    \setlength{\parskip}{6pt plus 2pt minus 1pt}}
}{
  \KOMAoptions{parskip=half}}
\makeatother
\usepackage{xcolor}
\makeatletter
\ifx\paragraph\undefined\else
  \let\oldparagraph\paragraph
  \renewcommand{\paragraph}{
    \@ifstar
      \xxxParagraphStar
      \xxxParagraphNoStar
  }
  \newcommand{\xxxParagraphStar}[1]{\oldparagraph*{#1}\mbox{}}
  \newcommand{\xxxParagraphNoStar}[1]{\oldparagraph{#1}\mbox{}}
\fi
\ifx\subparagraph\undefined\else
  \let\oldsubparagraph\subparagraph
  \renewcommand{\subparagraph}{
    \@ifstar
      \xxxSubParagraphStar
      \xxxSubParagraphNoStar
  }
  \newcommand{\xxxSubParagraphStar}[1]{\oldsubparagraph*{#1}\mbox{}}
  \newcommand{\xxxSubParagraphNoStar}[1]{\oldsubparagraph{#1}\mbox{}}
\fi
\makeatother

\providecommand{\tightlist}{%
  \setlength{\itemsep}{0pt}\setlength{\parskip}{0pt}}\usepackage{longtable,booktabs,array}
\usepackage{calc} 
\usepackage{etoolbox}
\makeatletter
\patchcmd\longtable{\par}{\if@noskipsec\mbox{}\fi\par}{}{}
\makeatother
\IfFileExists{footnotehyper.sty}{\usepackage{footnotehyper}}{\usepackage{footnote}}
\makesavenoteenv{longtable}
\usepackage{graphicx}
\makeatletter
\def\maxwidth{\ifdim\Gin@nat@width>\linewidth\linewidth\else\Gin@nat@width\fi}
\def\maxheight{\ifdim\Gin@nat@height>\textheight\textheight\else\Gin@nat@height\fi}
\makeatother
\setkeys{Gin}{width=\maxwidth,height=\maxheight,keepaspectratio}
\makeatletter
\def\fps@figure{htbp}
\makeatother

\makeatletter
\@ifpackageloaded{caption}{}{\usepackage{caption}}
\AtBeginDocument{%
\ifdefined\contentsname
  \renewcommand*\contentsname{Table of contents}
\else
  \newcommand\contentsname{Table of contents}
\fi
\ifdefined\listfigurename
  \renewcommand*\listfigurename{List of Figures}
\else
  \newcommand\listfigurename{List of Figures}
\fi
\ifdefined\listtablename
  \renewcommand*\listtablename{List of Tables}
\else
  \newcommand\listtablename{List of Tables}
\fi
\ifdefined\figurename
  \renewcommand*\figurename{Figure}
\else
  \newcommand\figurename{Figure}
\fi
\ifdefined\tablename
  \renewcommand*\tablename{Table}
\else
  \newcommand\tablename{Table}
\fi
}
\@ifpackageloaded{float}{}{\usepackage{float}}
\floatstyle{ruled}
\@ifundefined{c@chapter}{\newfloat{codelisting}{h}{lop}}{\newfloat{codelisting}{h}{lop}[chapter]}
\floatname{codelisting}{Listing}

\usepackage{amsthm}
\theoremstyle{plain}
\newtheorem{theorem}{Theorem}[section]
\theoremstyle{plain}
\newtheorem{proposition}{Proposition}[section]
\theoremstyle{definition}
\newtheorem{definition}{Definition}[section]
\theoremstyle{remark}
\AtBeginDocument{}

\makeatother
\makeatletter
\@ifpackageloaded{caption}{}{\usepackage{caption}}
\@ifpackageloaded{subcaption}{}{\usepackage{subcaption}}
\makeatother

\ifLuaTeX
  \usepackage{selnolig}  
\fi
\usepackage[]{natbib}
\usepackage{bookmark}

\IfFileExists{xurl.sty}{\usepackage{xurl}}{} 
\hypersetup{
  pdftitle={A Mathematical Framework and Software Implementation for Uncertainty Visualisation},
  pdfauthor={Harriet Mason; Dianne Cook; Sarah Goodwin; Susan VanderPlas},
  colorlinks=true,
  linkcolor={blue},
  filecolor={Maroon},
  citecolor={Blue},
  urlcolor={Blue},
  pdfcreator={LaTeX via pandoc}}

\begin{document}

\def\spacingset#1{\renewcommand{\baselinestretch}%
{#1}\small\normalsize} \spacingset{1}


\date{June 23, 2026}
\title{\bf A Mathematical Framework and Software Implementation for
Uncertainty Visualisation}
\author{
Harriet Mason\\
Department of Econometrics and Business Statistics, Monash University\\
and\\Dianne Cook\\
Department of Econometrics and Business Statistics, Monash University\\
and\\Sarah Goodwin\\
Department of Human Centred Computing, Monash University\\
and\\Susan VanderPlas\\
}
\maketitle

\bigskip
\bigskip
\begin{abstract}
Random variables are the bread and butter of statistics, and
visualisations are one of the most versatile tools in the field, so it
is a wonder why we do not have a methodology for visualising random
variables. This gap is particularly evident for exploratory data
analysis (EDA). We address this gap by designing a mathematical
framework for visualisation, which argues that we should consider
visualisations to be continuous functions. In the case of random
variable inputs, this means the visualisations should obey the
continuous mapping theorem (CMT). By breaking the visual function down
into its components, we are able to identify which parts of the mapping
are ill-defined for random variable inputs and redefine them in a way
that guarantees both the flexibility required for EDA and the
statistical sensibility of CMT. This formalisation represents a complete
integration of uncertainty into the grammar of graphics, which we show
by implementing the theory in the R package, ``ggdibbler''. The
ggdibbler software is a ``ggplot2'' extension that allows users to
replace the data of any plotting function with a random variable, with
the guarantee that the visualisation will have the same convergence
properties as its underlying data.
\end{abstract}

\newpage
\spacingset{1.9} 

\section{Introduction}\label{introduction}

\begin{quote}
``Elegant design requires us to think about a theory of graphics, not
charts.'' \citet{Leland2005}
\end{quote}

Uncertainty visualisation has suffered from a serious lack of
formalisation in recent years, which has resulted in an overwhelming
tsunami of named plots, touching every corner of the literature. No
visualisation problem is too niche as we see visualisations of
2D-intervals called a cross plot, rectangle plot, segment plot, and
dandelion plot \citep{Zhang2022}. No plot change is too small with
simulated outcomes on lines called spaghetti plots (or lasagne plots for
gradients) \citep{Swihart2010}, on maps called pixel maps
\citep{Vizumap}, on bar charts called fuzzygrams \citep{Kay2023}, and
when the outcomes are animated, we call them hypothetical outcome plots
(HOPs) \citep{Hullman2015}. No dead horse is too beaten with probability
functions having more names than ways to differentiate them with terms
like histogram, density plot, violin plot, ridgeline plot, rain plot,
dot plot, and gradient plot, only scratching the surface
\citep{Kay2023}. The practice of naming plots turns the field into a
race, where authors are incentivised to brand their name on as many
graphics as possible, rather than work towards a cohesive theory of
visualisation. This problem is as pervasive as it is frustrating.

According to \citet{Leland2005}, ``the computer woefully focuses the
mind'', so we might assume that these hyper-specific naming conventions
do not extend to the landscape of uncertainty visualisation software.
This assumption would be wrong. Even just within the R ecosystem of
\texttt{ggplot2} extensions, there is an overwhelming number of choices,
leaving users unclear when each package should be used and for what
purposes. A density plot can be made using \texttt{ggplot2},
\texttt{ggbeeswarm}, \texttt{ggridges}, \texttt{ggrain},
\texttt{ggdist}, \texttt{ggpointdensity}, and \texttt{ggdensity}, all
packages with similar descriptions and overlapping designs. This is not
to say these packages are pointless: they have thousands of downloads
per week and were motivated by a gap in the literature. Rather, this
example illustrates the havoc that named plots contribute to our
software ecosystem.

With this wealth of choices in plots and software, it would be
reasonable to assume that the field is meeting the needs of its users,
albeit in a messy, roundabout way, but this assumption would be
incorrect. Despite this wealth of choice, each package is brittle and
constrained by its own set of assumptions. There is a cry for
flexibility among the users of uncertainty visualisation software,
dating back more than 30 years \citep{MacEachren1992}. A useful
uncertainty visualisation system should allow for exploration with
uncertain data from multiple different sources, dimensions, and data
types \citep{Hadjimichael2024, geointerviews, MacEachren2005}, a
flexibility that all current implementations fail to provide.

This implies the real gap in uncertainty visualisation is not new plots
or bespoke software, but structure, and the flexibility that comes with
it. Structure is not so alien to the visualisation community that this
is an unreasonable request. Standard statistical graphics have been
formalised within \textbf{the grammar of graphics}
\citep{Leland2005, ggplot2} for several decades. This structure has
largely been ignored by the uncertainty visualisation community, with
the exception of \citet{Kay2023} and the \texttt{ggdist} R package.
Along with this formalisation of visualisations of univariate
probability functions, Kay expressed hope for a single coherent
uncertainty visualisation framework. This single unified framework, and
its implementation in the \texttt{ggdibbler} package, is what will be
discussed in the rest of this paper.

\section{A motivating example}\label{a-motivating-example}

We are going to build our uncertainty visualisation system upon a simple
assumption based on the existing literature: the input for an
uncertainty visualisation is a data set where every cell is a random
variable \citep{Kay2023} which contains all the quantified uncertainty
we wish to represent \citep{Mason2026}. Once we have our distribution
inputs, there is only one question left for our uncertainty
visualisation system to solve. What exactly should it \emph{do} with
these distributions?

Consider the case of a vector, \(X\), that contains the heights of 15
different women, where each \(x_i \in \mathbb{R}\) represents one
measured height. If we feed this vector into a density plot function,
such as \texttt{geom\_density()} from \texttt{ggplot2}, how should it
behave? Hopefully, this case is obvious, and the output of this function
is shown in the \texttt{ggplot2} example in
Figure~\ref{fig-dist-example}. A density curve of our data is estimated
and displayed using a line geometry. This is a straightforward example,
but what happens when our input is a distribution? This can occur if,
for example, the tool used to record the heights has an inherent
measurement error. Now, we have a vector, \(\textbf{X}\), of 15
distributions, where each \(\textbf{x}_i \sim N(x_i, \sigma_i)\), where
\(x_i \in X\) is the recorded value, and \(\sigma_i\) is an estimated
variance based on the environmental conditions of the measurement. There
are two routes we can take when it comes to visualising this estimate.
The second plot in Figure~\ref{fig-dist-example} is made using
\texttt{ggdist}, and it shows the case where we are interested in the
distribution of each individual measurement. The third plot is made
using \texttt{ggdibbler}, which places the emphasis on the full data
density \emph{but} carries forward the variability in the density that
comes with having distributional inputs.

\begin{figure}[t]

\centering{

\includegraphics[width=1\textwidth,height=\textheight]{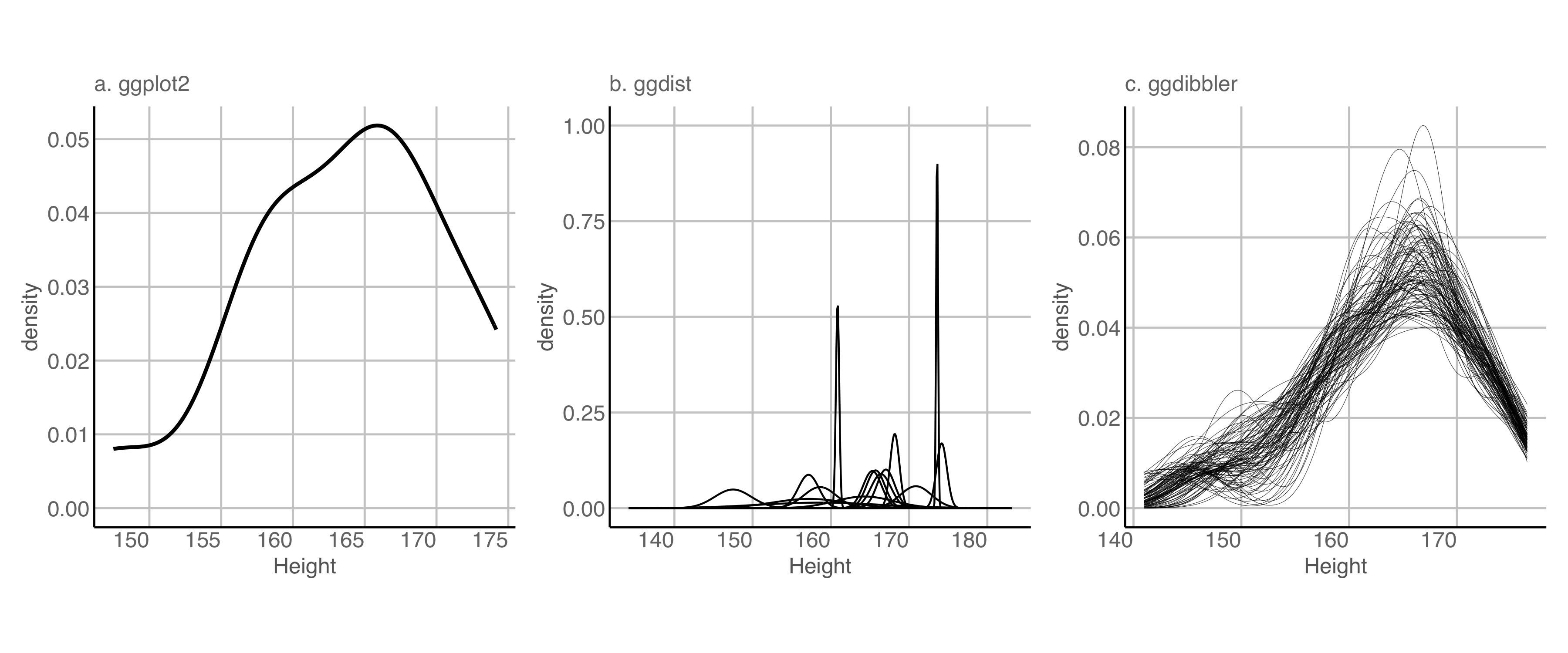}

}

\caption{\label{fig-dist-example}Alternative interpretations of how to
render a density plot when the input is a set of distributions
describing uncertainty of measurements, according to three plotting
packages: (a) \texttt{ggplot2} forms the density from the mean values,
(b) \texttt{ggdist} puts the distribution on each observation, treating
uncertainty as signal, (c) \texttt{ggdibbler} shows the densities for
multiple samples, which puts the focus on how the density might look
given the uncertainty. These differences illustrate how uncertainty is
interpreted in different ways. Which is correct?}

\end{figure}%

The distinction between the two approaches presented in
Figure~\ref{fig-dist-example} is the same signal and noise paradigm
presented by \citet{Mason2026}. In the \texttt{ggdist} plot, we are
interested in the shape of the distribution of each observation, so we
are visualising the uncertainty as a signal. In the \texttt{ggdibbler}
plot, we are not interested in the uncertainty in and of itself, but
rather, we only included it to see how it would change the conclusions
from the \texttt{ggplot2} plot, thus, visualising it as noise. Given
these two approaches, which plot is the ``correct'' visualisation
depends on the goals of our analysis and what we are looking to infer
from making the plot. However, this freedom to choose disappears if we
want to design a visualisation system for EDA. If we want to design a
visualisation system for EDA, the method must always work, which means
any \emph{existing} plot should always have an ``uncertain''
counterpart. Therefore, when designing a system for EDA, we are actually
interested in a slightly different question: ``Which plot is the correct
behaviour of \texttt{geom\_density} with the distributional input,
\(\textbf{X}\)?''. The answer to this question is definitively, the
\texttt{ggdibbler} plot. The reasoning is rooted in strong statistical
foundations; we just need to expand the concepts to visual statistics.

\section{Visual Statistics}\label{visual-statistics}

\subsection{The deterministic visual
function}\label{the-deterministic-visual-function}

When you really think about it, what is a visualisation? One of the most
definitive answers to this question can be found in the grammar of
graphics, which is a theoretical framework that characterises a
visualisation as a series of composite functions \citep{Leland2005}.
Under this formalisation, a visualisation is a function that takes a
deterministic matrix (Definition~\ref{def-deterministic}) as input and
outputs a statistical graphic. This formalisation is lengthy, so we
summarise the key aspects of it in Definition~\ref{def-vfunc}, and
Figure~\ref{fig-gog}, which shows each step of the visualisation process
with highlighted components indicating the sections we will need to
adjust for uncertainty visualisation. By leveraging the graphics
pipeline, we make the implicit explicit, which will allow us to compare
different uncertainty visualisation systems on a deeper and more
objective level \citep{Wickham2009}. Note that these steps summarise the
underlying architecture of the \texttt{ggplot2} \citep{ggplot2}
implementation of the grammar, which is an essential foundation for this
new software. However, the conceptual framework described here is not
dependent on any particular implementation of the grammar and could
easily be implemented in another grammar of graphics system, such as
Vega-Lite \citep{Satyanarayan2016}. The notation used for visual
statistics builds on that developed in \citet{Majumder2013} where the
term was first introduced in the context of conducting statistical
inference.

\begin{definition}[Deterministic
matrix]\protect\hypertarget{def-deterministic}{}\label{def-deterministic}

Let \(\textbf{A}\) be an \(m \times n\) matrix on the sample space
\(\Omega\). A is a deterministic matrix if \(a_{i,j} \forall i,j\) are
single-value (non-random) entries.

\end{definition}

\begin{definition}[The deterministic visual
function]\protect\hypertarget{def-vfunc}{}\label{def-vfunc}

Let \(X_{n_1,m_1}\) be a matrix of outcomes with \(n_1\) rows and
\(m_1\) columns on the sample space \(\Omega\). Let \(V\) be a function
that maps \(X\) from our sample space \(\Omega\) to the space of all
visual statistics, \(\Psi\). We refer to \(V\) as a visual function, and
its output \(V(.)\) as a visual statistic. The visual function, \(V\),
can be decomposed into the following composite function:

\[
V = A \circ O \circ G \circ S \circ M
\] where:

\begin{itemize}
\tightlist
\item
  \(M = [M_1, M_2, ..., M_{m_1}]'\) is a set of \(m_1\) functions that
  each scale one column of our \(X_{n_1,m_1}\) matrix. This function
  maps the individual cells of \(X_{n_1,m_1}\), \(x_{i,j}\)
  \(\forall i = 1,...,n_1\) and \(j = 1,...,m_1\) from the sample space,
  \(\Omega\), to the space of real numbers, \(R\).
\item
  \(S\) is a statistic function that summarises our \(X_{n_1,m_1}\)
  matrix down to an \(X_{n_2,m_2}\) matrix. There are no strict
  requirements for the statistic function. The function is a
  transformation from \(R\) to \(R\).
\item
  \(G\) converts two (for a 2D graph) position columns
  \(X_{[.,k]},X_{[.,l]}\) to values that represent magnitudes in space
  returning an \(X_{n_2,m_2}\) matrix on the bounded plane \(B^2\). We
  can further decompose \(G\) into \(G = E\circ P\), where is the
  geometry function that maps each \(x_{i,j} \forall j=k,l\) from \(R\)
  to the bounded plane \(B^2\), and \(P\) is a position modifier
  function that checks for overlapping values in \(X_{n_2,m_2}\) and
  either stacks them on top of each other in the dependent axis, or
  dodges them next to each other on the independent axis.
\item
  \(O\) is a transformation of our coordinate system. The function is a
  transformation from \(B^2\) to \(B^2\).
\item
  \(A\) is an injective function that transforms our \(X_{n_2,m_2}\)
  matrix into physically observable stimuli. The function maps our data
  from \(B\) to the space of visual statistics \(\Psi\).
\end{itemize}

\end{definition}

\begin{figure}[t]

\centering{

\includegraphics{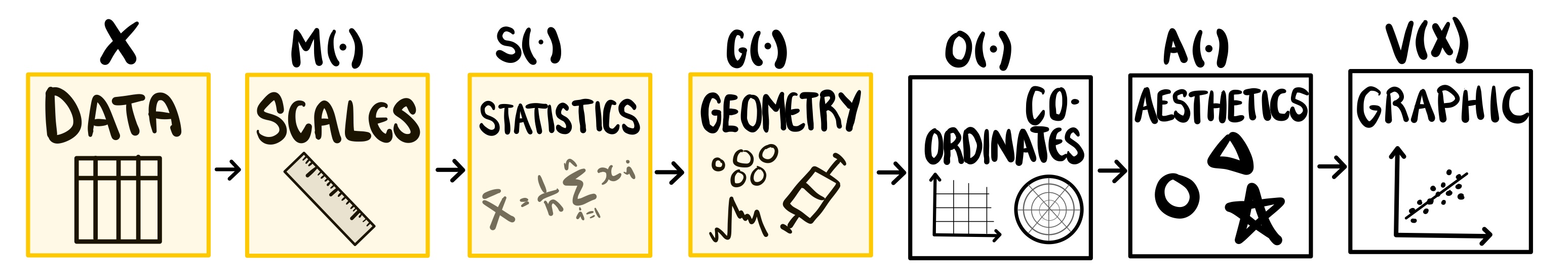}

}

\caption{\label{fig-gog}Illustration of the steps to render a plot, as
defined by the grammar of graphics. For uncertainty visualisation,
changes are needed for the highlighted steps: data, scales, statistics,
and the position adjustment within the geometry.}

\end{figure}%

\subsection{Random matrices and continuous mapping
theorem}\label{random-matrices-and-continuous-mapping-theorem}

Unlike a standard visualisation, an uncertainty visualisation has a
random matrix (Definition~\ref{def-random}) input, where distributions
replace the single values of our deterministic matrix
(Definition~\ref{def-deterministic}).

\begin{definition}[Random
matrix]\protect\hypertarget{def-random}{}\label{def-random}

Let \(\textbf{A}\) be an \(m \times n\) matrix-valued random variable on
the probability space \((\Omega, \mathcal{F}, Pr)\). This is a matrix
where some or all of its entries are random variables drawn from some
probability distribution.

\end{definition}

Combined, Definition~\ref{def-random} and Definition~\ref{def-vfunc}
boil uncertainty visualisation down to a simple problem. Our
visualisation system needs to be designed in such a way that it follows
existing mathematical principles of random variables and functions. More
specifically, our graphics should uphold the \emph{continuous mapping
theorem} \citep{mann_wald}, Theorem~\ref{thm-cmt}.

\begin{theorem}[Continuous mapping
theorem]\protect\hypertarget{thm-cmt}{}\label{thm-cmt}

Let \(\textbf{X}\) and \(\textbf{Y}\) be \(n \times m\) random matrices,
and let \(Z\) be an \(n \times m\) deterministic matrix. Let
\(f: E\xrightarrow{} E'\) be a continuous function from one metric
space, \(E\), to another \(E'\). Then:
\[\textbf{X} \xrightarrow{} \textbf{Y} \Rightarrow f(\textbf{X}) \xrightarrow{} f(\textbf{Y})\]

and

\[\textbf{X} \xrightarrow{} Z \Rightarrow f(\textbf{X}) \xrightarrow{} f(Z)\]

\end{theorem}

In simple terms, Definition~\ref{def-vfunc} and Theorem~\ref{thm-cmt}
mean that our uncertainty visualisations should have similar convergence
properties to their underlying distributions.

If we accept the idea that a visualisation is a continuous function,
adhering to Theorem~\ref{thm-cmt} is not a nice property or an opinion
on how visualisations should behave, but rather a \emph{fundamental
requirement} of any visualisation of a random matrix. Or, at least it
will be after we establish that the assumptions of Theorem~\ref{thm-cmt}
are true. This will require us to show that \(\Omega\) and \(\Psi\) are
metric spaces, and then use those metric spaces to define convergence
for visual functions.

\subsubsection{Visual metric spaces}\label{visual-metric-spaces}

The idea that both \(\Omega\) and \(\Psi\) are metric spaces is not
particularly strange: one of the core tenets of statistical graphics is
that they maintain the link between data and visual aesthetic
\citep{Leland2005}. This point was made rather comically by
\citet{Bartonicek2025}, who drew two rectangles stacked on top of each
other, and pointed out that it was not a stacked bar chart. The idea
that there is some kind of structure, or ordering that we need to
maintain, implies that \(\Omega\) is a metric space. In most cases,
\(\Omega \subseteq \mathbb{R}^p\) and it immediately follows that the
ordered pair (\(\Omega\), \(d\)) is a metric space, where \(d\) is
Euclidean distance. In the cases where
\(\Omega \not\subseteq \mathbb{R}^p\), the first step of the analysis is
to scale the data to \(\mathbb{R}^p\) using \(M\), so we can say that
(\(\Omega\), \(d \circ M\)) is a metric space. For our visual space,
\(\Psi\), we can actually do exactly the same thing, but in reverse.
Since our aesthetic function is defined as \(A: B\xrightarrow{} \Psi\),
and \(B \subseteq \mathbb{R}^p\), we can set our metric to be the
inverse of our aesthetic function, such that the ordered pair (\(\Psi\),
\(d \circ A^{-1}\)) is a metric space.

Ideally, as visualisations are designed to be viewed by humans, we would
have set our distance on \(\Psi\) to be the human ability to visually
differentiate two plots, \(h\). This is similar to the notion of
distance that is leveraged by the lineup protocol \citep{Buja2009}. In
the lineup protocol, viewers are shown \(M-1\) null plots,
\(V(X_{N_i})\), where \(X_{N_i}\) for \(i=1, ..., M-1\) are independent
draws generated from some null distribution, and a visual test
statistic, \(V(X_T)\), where \(X_T\) is our actual data. If the test
statistic, \(V(X_T)\), is significantly different, visually, from
\(V(X_{N_i}) \forall i\), then viewers will be able to pick it out of
the lineup, and we would reject our null hypothesis that \(X_T\) was
generated from the same distribution as \(X_{N_i}\). This test actually
measures Mahalanobis distance on \(\Psi\), as Mahalanobis distance is a
multivariate generalisation of a Z score, and the lineup protocol is the
visual equivalent of a hypothesis test. Even though human perception is
the natural metric for statistical graphics, (\(\Psi\), \(h\)) it is not
a metric space. This is because it violates the triangular inequality
due to the existence of ``just noticeable differences'' (JND)
\citep{Luce1958}. We can show that human perception is not a metric
space with a quick proof.

\begin{proposition}[]\protect\hypertarget{prp-metric}{}\label{prp-metric}

Let \(h: \Psi \xrightarrow{} \mathbb{R}\) be a piecewise distance
function that measures the human perception of statistical graphics,
defined as: \[
h(x, y) = \begin{cases} 
      g(x,y) & g(x,y) \geq \epsilon \\
      0 &  g(x,y) < \epsilon
   \end{cases}
\] where \(\epsilon\) is the JND of our human observer (with
\(\epsilon>0\)), and \(g(x,y) = d(A^{-1}(x), A^{-1}(y))\) is the
Euclidean distance in the rendered graphics; \(h\) is not a metric.

\end{proposition}

\begin{proof}
Let \(a,b,c \in \Psi\) be three different visualisations, where \(a\)
and \(c\) are plots in the metric space, and \(b\) is exactly halfway
between them, such that \(g(a, c) =  \epsilon\),
\(g(a,b) = \frac {1}{2} \epsilon\), and
\(g(b, c) = \frac {1}{2} \epsilon\). Therefore, \(h(a, c) =  \epsilon\),
\(h(a,b) = 0\), and \(h(b, c) = 0\). If we assume \(h\) is a metric
space, then the triangular inequality will hold, and we can state
\[h(a, c) \leq h(a,b) + h(b, c)\] which implies \[0 \leq \epsilon\]
Therefore, by contradiction, \(h\) is not a metric.
\end{proof}

As long as our graphics system does not let \(A\) map to increments
\(< \epsilon\), this should not be a problem, but due to differences in
human perception, \(\epsilon\) is not constant for the entire human
population, and this system would be impossible to implement. This means
that we might not be able to visually distinguish every plot that is
different in \(\Psi\), but if we are able to \emph{see} a difference in
two plots, it is definitely there.

\subsubsection{Visual convergence}\label{visual-convergence}

With the metric space out of the way, we need to define the concept of
convergence for visual functions. Thankfully, this is also covered by
our definition of a metric space. Two plots have converged if their
renderings are identical.

\begin{definition}[Visual
convergence]\protect\hypertarget{def-converge}{}\label{def-converge}

Let \(\textbf{X}\) and \(\textbf{Y}\), be \(n \times m\) random
matrices, and let \(Z\) be an \(n \times m\) deterministic matrix. Let
\(V\) be a visual function \(V\): \(\Omega\xrightarrow{} \Psi\). Let
\(g\): \(\Psi\xrightarrow{} R\) where \(g = d \circ A^{-1}\), and
\((\Psi, g)\) is a metric space. We say two random graphics have
visually converged, that is,
\(V(\textbf{X}) \xrightarrow{} V(\textbf{Y})\) when
\(g(V(\textbf{X}), V(\textbf{Y})) = 0\). We say that a random graphic
has converged to a deterministic graphic, that is,
\(V(\textbf{X}) \xrightarrow{} V(X)\) when
\(g(V(\textbf{X}), V(X)) = 0\).

\end{definition}

We will usually approximate this convergence by using visual
distinguishability, similar to the approach taken by the lineup
protocol.

\subsection{Returning to the density plot
example}\label{returning-to-the-density-plot-example}

In our density example, we defined \(\textbf{X}\) and \(X\) such that
\(\textbf{x}_i \xrightarrow{p} x_i, \forall i = 1,...,15\) as
\(var(\textbf{x}_i) \xrightarrow{} 0\). Therefore, as the variance of
our distributions approach zero, \(\textbf{X}\xrightarrow{p}X\) and we
should see \(V(\textbf{X})\xrightarrow{p}V(X)\). That is, as the
variance of all the heights in \(\textbf{X}\) approaches zero, our
uncertainty visualisation should be visually indistinguishable from the
\texttt{ggplot2} plot of Figure~\ref{fig-dist-example}. Looking at the
plots, we would observe this behaviour in the \texttt{ggdibbler} plot,
but not the \texttt{ggdist} plot. This is why we assert that the
\texttt{ggdibbler} plot is the random matrix version of the
\texttt{geom\_density} function from \texttt{ggplot2}.

\section{Generalising the visual
function}\label{generalising-the-visual-function}

The visual function, as described in the \emph{grammar of graphics}, has
one key limitation; it assumes each data point is deterministic. That
is, Definition~\ref{def-vfunc} in its current form does not allow for
random matrix input. This is an issue for uncertainty visualisation, as
a random matrix input is one of the core assumptions of the approach. To
redefine our visual function such that it accepts random matrix inputs,
we will need to adjust the definition of our \textbf{scale},
\textbf{statistic}, and \textbf{geometry} to generalise
Definition~\ref{def-vfunc}. In doing so, we will also ensure we do not
make any changes that will result in a violation to
Theorem~\ref{thm-cmt}. When we finish, we should have a generalised
visual function that will accept random matrix inputs for any graphic
created using the \emph{grammar of graphics}, with the additional
property that the plots always obey Theorem~\ref{thm-cmt}.

\subsection{\texorpdfstring{The adjustment to
\texttt{scales}}{The adjustment to scales}}\label{the-adjustment-to-scales}

Scales map our data to a set of real number outcomes; they determine how
we perceive the size, shape, and location of our data, and give our data
meaning \citep{Leland2005}. This sentiment is echoed in the construction
of our visual statistics, as the scale \(M\) is used to define the
distance metric for our metric space. To use the same scale, \(M\), that
was defined for our deterministic matrix \(X\) on our random matrix
\(\textbf{X}\), we will need to perform a ``change of variable''.
\citet{Kay2023} would refer to this as a ``scale aware'' requirement for
uncertainty visualisation systems. This means we are not changing the
function \(M\) itself, but we are just expanding the definition to allow
for random matrix input. This formalisation of this scale is summarised
in Definition~\ref{def-scale}, and Figure~\ref{fig-scale},

\begin{definition}[Generalised
scale]\protect\hypertarget{def-scale}{}\label{def-scale}

Let \(\Omega\) be a sample space and \(M\): \(\Omega\xrightarrow{} R\)
be a scale function that maps our data from \(\Omega\) to \(R\). Let
\(\textbf{X}\) be a random matrix on the probability space
\((\Omega, \mathcal{F}, P)\). Then the scale function \(M\), applied to
\(\textbf{X}\) will give us \(M(\textbf{X})\) with the induced
probability measure \(P_{M(X)}(A) = P(M^{-1}(A))\).

\end{definition}

\begin{figure}

\centering{

\includegraphics[width=0.8\textwidth,height=\textheight]{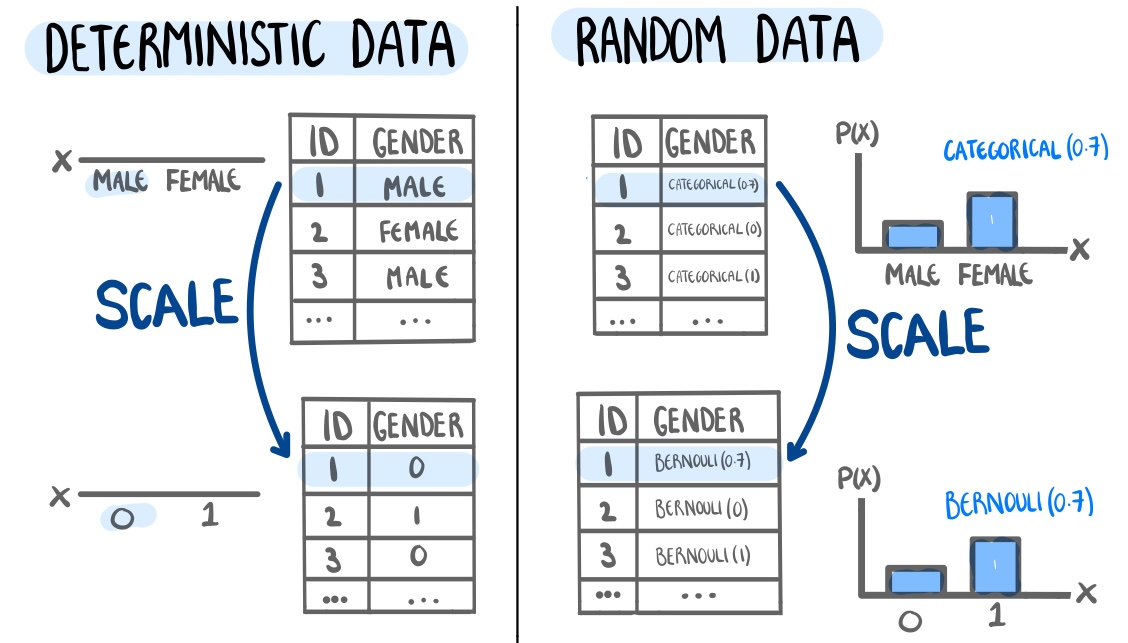}

}

\caption{\label{fig-scale}An illustration that highlights the difference
between scaling distributions and scaling deterministic variables.
Scaling deterministic variables only requires us to map individual
values, but scaling a distribution requires us to scale the
distribution's domain.}

\end{figure}%

The illustration in Figure~\ref{fig-scale} shows that our distribution
can also changes names (i.e.~we scale a binary categorical random
variable to \{0,1\} to create a Bernoulli distribution), but this is
simply a function of the scaling, and does not represent any meaningful
change in the shape of the distributions.

\subsection{\texorpdfstring{The adjustment to
\texttt{statistics}}{The adjustment to statistics}}\label{the-adjustment-to-statistics}

Statistics provide a summary of our data. Several graphics, such as box
plots or bar charts, are inherently linked to a statistic (the
five-number summary and summation, respectively). This is the current
role of our statistic function, \(S\), which is only well defined for
deterministic inputs. Unlike deterministic variables, random variables
are more abstract and cannot be expressed as a single concrete value.
Therefore, there is a second statistic that needs to be computed, which
is the statistic that represents the distribution. The distributional
statistic must be calculated at this stage, as the geometry assumes we
are working with concrete values that can be mapped to a position in
Euclidean space. Therefore, in uncertainty visualisation, there are two
statistics that must be defined: the statistic that represents the
distribution, and the statistic that summarises the data. The
generalised statistic function that accepts random variables is shown in
Definition~\ref{def-stat}.

\begin{definition}[Generalised
statistic]\protect\hypertarget{def-stat}{}\label{def-stat}

Let \(\textbf{X}_{n_1,m_1}\) be a random matrix on the probability space
\((R, \mathcal{F}, P)\). Let
\(S_{sample}: (R, \mathcal{F}, P)\xrightarrow{} R\) be a function that
transforms the random matrix \(\textbf{X}_{n_1,m_1}\) into the
deterministic \(X_{n_1,m_1, t}\) array, where \(t\) is the number of
samples drawn from \(\textbf{X}_{n_1,m_1}\). Let \(S\):
\(R\xrightarrow{} R\) be a function that transforms the deterministic
matrix, \(X_{n_1,m_1, t}\) into a statistical summary,
\(X_{n_2,m_2, t}\). Note that \(S\) covers all statistics that can be
implemented in the deterministic grammar of graphics. We define
\(S^*: (R, \mathcal{F}, P)\xrightarrow{} R\) to be the composite
function \(S = S_{sample} \circ S^*\) that transforms the random matrix
\(\textbf{X}_{n_1,m_1}\) into the deterministic matrix,
\(X_{n_2,m_2, t}\).

\end{definition}

The definition presented in Definition~\ref{def-stat} may leave you with
some questions, specifically, why we have limited the distribution
representation to a sample of outcomes. This is the question we will
spend the rest of this section answering.

\subsubsection{Representing a
distribution}\label{representing-a-distribution}

Distributions are abstract concepts, needing some kind of concrete
representation to visualise them. This is a common consideration in
uncertainty visualisation where we frequently see software that lets us
visualise our distribution as a sample of outcomes, a mean and variance,
a confidence interval, or even several of these statistics at once
\citep{Potter2010, Kay2023}. Authors often opt for flexibility in the
distribution level statistic, with very little consideration as to how
this might affect the rest of the plot. It is not unreasonable to assume
that the distribution statistic should be interchangeable, as this is
how the \emph{rest} of the grammar operates. Even \citet{Leland2005}
himself expressed this sentiment in the uncertainty visualisation
chapter of The Grammar of Graphics. What this sentiment fails to realise
is that the interchangeable nature of the grammar comes from its
formalisation, a formalisation that does not properly integrate
uncertainty. When we take the time to formalise uncertainty within the
visualisation function, we quickly see that the way we represent our
distribution does not have a neutral impact on the other components of
the grammar.

\subsubsection{Beyond point estimates}\label{beyond-point-estimates}

The first problem is obvious: not all distribution representations can
show the uncertainty in our random matrix. We briefly illustrate the
problem here. Figure~\ref{fig-meanprob} shows a set of bivariate
densities as raster plots visualising the \texttt{uncertain\_faithfuld}
data from \texttt{ggdibbler}, which is a random matrix variation of the
\texttt{faithfuld} data from the \texttt{ggplot2} package. Plot (a)
visualises only the estimate, while the other three plots represent the
data using a sample. As we move from plot (b) to plot (d), the variance
in our distribution increases. This variance is independent of our
expected value, so the increasing variance has no impact on the expected
value of the distribution. Unlike the sample visualisation, the
visualisation of our point estimate does not change as the variance
increases and is always represented by plot (a), which is equivalent to
always conveying a variance of zero. The example shows how
\texttt{ggdibbler} handles increasing variance in a density plot, as the
``graininess'' of the plot increases with the uncertainty. As the
variance increases, these grains dominate the plot, making the
visualisation harder to read, as it should be.

\begin{figure}[t]

\centering{

\includegraphics[width=1\textwidth,height=\textheight]{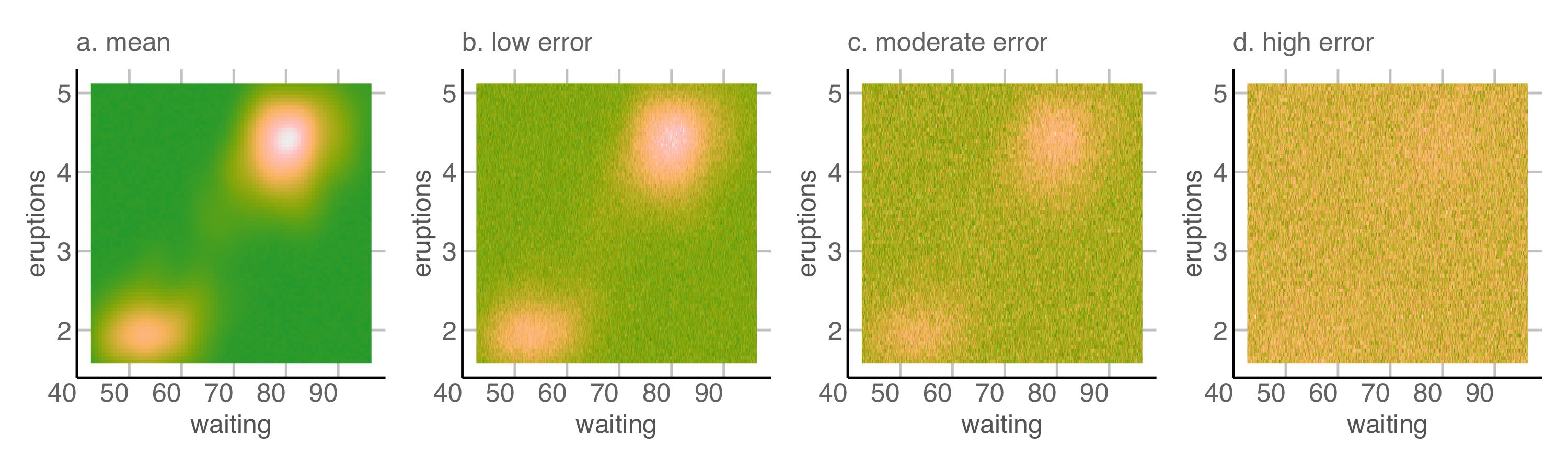}

}

\caption{\label{fig-meanprob}How uncertainty is handled (or not) in
raster displays of bivariate density. The axes show the eruption time vs
waiting time, and colour indicates density value, with lighter
indicating higher density. In plot (a), uncertainty is ignored by
showing only the estimate, and plots (b, c, d) show samples reflecting
different scales of uncertainty in the density estimate. We can see that
as the variance in the estimates increases, the visualisation of the
sample becomes harder to read and conveys more uncertainty.}

\end{figure}%

You may still want to visualise summary statistics alongside our
uncertainty representation. This is a common sentiment, and it is the
same desire that causes us to visualise a mean and confidence interval
on the same plot. Unfortunately, there is a reasonable amount of
evidence that visualising summary statistics alongside uncertainty
information causes that uncertainty information to be ignored
\citep{uncertchap2022}. Allowing you to include the point estimates that
allow you to ignore the noise, explicitly because the visualisation is
too noisy, defeats the entire purpose of the approach. If we are only
concerned with convergence to constant values (and not visual
convergence to other distributions), then we do not need to include
uncertainty at all and Theorem~\ref{thm-cmt} would be trivially
fulfilled by visualisations of the mean. Therefore, whichever
representation we choose, it needs to convey a complete view of this
distribution; a point estimate cannot do that.

\subsubsection{Why not probability
functions?}\label{why-not-probability-functions}

Disallowing point estimates doesn't actually limit our flexibility, as
we still have samples, quantiles, and probability functions at our
disposal. This is where the ``Mr Potato Head'' approach to
distributional statistics starts to cause problems as we bump up against
the orthogonality requirement that is built into the grammar of
graphics. Flexibility requires that \emph{every} deterministic graphic
should have an uncertain counterpart, even ones that have a pre-defined
statistic comprised of point estimates, such as a box plot or bar chart.
This is a fundamental requirement, otherwise the system will not be an
effective EDA tool. If we allow for any statistic, we will create a
mismatch where the values we are trying to feed into our statistic,
\(S\), are not on the same domain as expected by the function. To be
more explicit, our statistic is expecting values on the domain,
\(M(\Omega)\). If we define a new statistic, \(S^* = S \circ S_{dist}\),
where the range of \(S_{dist}\) is not \(M(\Omega)\), such as
\(P_{M(\Omega)}(M(\Omega))\), then we have produced invalid input for
the next stage of our visual function, \(S\). For example, if our
statistic is expecting heights that range from 150 to 200, we cannot
feed in a set of probabilities on {[}0,1{]} and expect there to be no
issues. The statistics that create a domain mismatch also tend to create
an implicit inference problem, as changing the statistic used to
represent the random variables can also change the role of uncertainty
in our analysis \citep{Mason2026}. This means that violating this rule
will not only result in unusable inputs or nonsensical outputs for
\(S\), it will also fundamentally change our visual function, \(V\).
This change means our visualisation will not adhere to
Theorem~\ref{thm-cmt}, which is the primary requirement of our system.
For these reasons, we can only allow statistical representations that
output a range that is equivalent to the input space.

\subsubsection{Why not quantiles?}\label{why-not-quantiles}

This leaves quantiles and samples as our remaining distribution
statistics. It makes sense that these methods would be the most
flexible, as a sample is just outcomes on \(M(\Omega)\), and quantiles
are just ordered samples. However, the notion of ``ordering'', which is
required for quantile representations, produces two problems for a
flexible visualisation system. The first problem is that quantiles
communicate an explicit ordering on \(\Omega\). While the data at this
stage is technically on the real line \(M(\Omega)\), the quantiles will
not have meaning if \(\Omega\) is unordered. Using quantiles to
visualise uncertain categorical data will either result in meaningless
graphics or an inability to visualise the data at all. This limitation
would prevent us from visualising uncertain categorical data, which is a
common output of classification models. The second problem with
quantiles is that they don't have a natural extension to multivariate
data. Quantiles are well-defined for univariate cases, but multivariate
spaces require several assumptions on the relative magnitude of our
variables, which are unlikely to always be correct.
Figure~\ref{fig-circle-line} visualises the four scenarios that arise
from passing a univariate or multivariate random variable, represented
as a quantile or a sample, to the slope or intercept of a
\texttt{geom\_abline}. This is not an unreasonable scenario, as a linear
regression with a random intercept and slope is a common topic even in
introductory statistics courses. We can see that in the univariate case,
where the intercept of the line is an \(N(0,1)\) distribution, the
information conveyed by the quantile (a) and the sample (b) is very
similar. This is because quantiles are well defined in the univariate
case. However, for the multivariate case, adding a second random
variable in the slope (such that we are now visualising a multivariate
normal distribution with marginal distributions \(N(0,1)\), and a
covariance of \(-0.8\)) throws our notion of ordering out the window. In
this case, quantiles are not well defined, as any quantile, \(q\), will
have an infinite number of (intercept, slope) pairs that could produce
that probability, and are better conveyed by a function (hence why
contour plots are typically used). As we cannot visualise the slope of a
line as a function, a sensible alternative might be to use the marginal
or equicoordinate quantiles \citep{Bornkamp2018}. We opted to use the
marginal quantiles, which are also used to colour the lines in both
representations, but the conclusions are the same if we use the
equicoordinate approach instead. This allows us to see the danger of
using quantiles as our distribution representation. The first issue is
that the notion of ordering we have imposed on the quantiles does not
translate in the multivariate case, which can be seen in the haphazard
colouring of the lines in the sample, which was not a problem in the
univariate case. The implicit pairing of values has also changed the
point of intersection of the lines, and the neatness of the quantiles
conveys more certainty in our conclusions than is warranted by the
actual data. Even if we tried to work around these problems by coming up
with some abstract definition of a visual quantile, we wouldn't be able
to draw the output with a straight line, which is the only real
requirement for \texttt{geom\_abline}.

\begin{figure}

\centering{

\includegraphics[width=1\textwidth,height=\textheight]{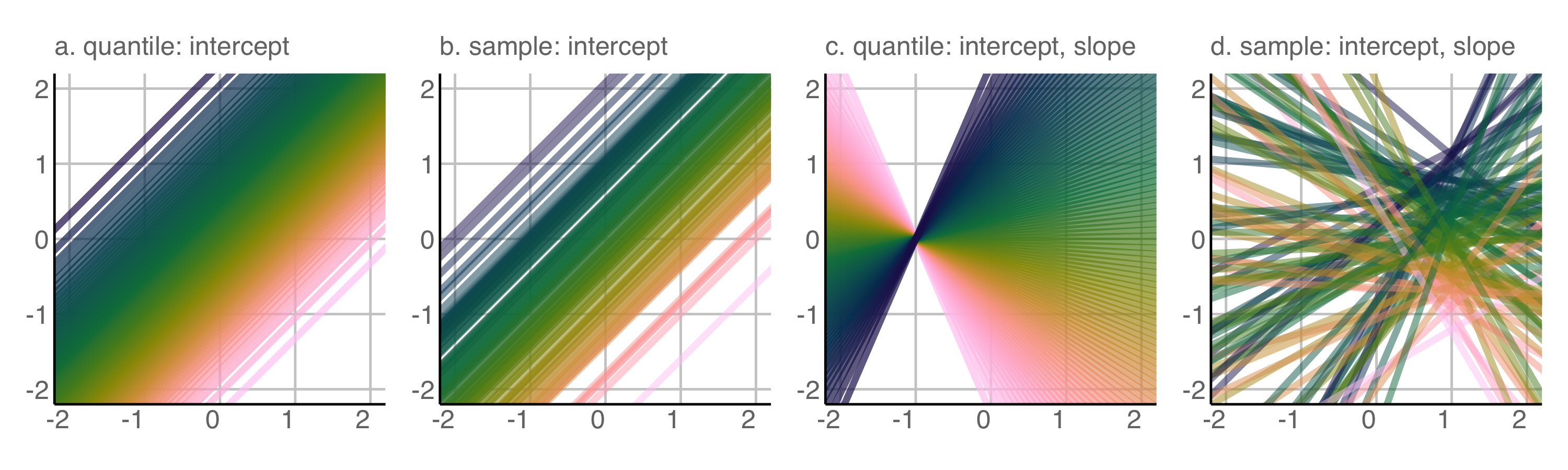}

}

\caption{\label{fig-circle-line}Why quantiles are problematic for
representing our distribution variables, using regression coefficients.
Intercepts and slopes were simulated using marginal distributions of
N(0,1) and a covariance of -0.8. Plots (a) and (b) have only the
intercept treated as random, and show the quantiles and samples,
respectively. Colour maps to quantile in (a) and to the value of the
intercept in (b): both plots convey similar information. It breaks down
when both slope and intercept are treated as random, shown as quantiles
(c) and samples (d). Colour is mapped to the same notion of distance
that is used to construct the quantiles. But distance is not well
defined, and we can see the approaches diverge in the erratic colouring
of the lines. Ordering beyond one variable makes quantiles inflexible
representations of distributions.}

\end{figure}%

\subsubsection{Distributions as samples}\label{distributions-as-samples}

This means that the only representation of a distribution that is
equally as flexible as a point prediction is a sample of outcomes. Of
course, we don't want a sample of individual points; we want a sample of
geometric objects. To get this, we need to pass the data through the
visual function in batches, where each batch represents an outcome of
the full random matrix. In practice, this translates to ``splitting'' on
the \texttt{drawID} in the Grammar of Graphics \citep{Leland2005}, which
involves modifying the \texttt{group} variable in \texttt{ggplot2}
\citep{ggplot2} to include the \texttt{drawID}.
Figure~\ref{fig-grouping-need} demonstrates the effect of grouping in
the implementation. Without appropriately handling the grouping in
\texttt{geom\_smooth}, the plot has only one fitted curve with a
standard error artificially small due to a larger number of
observations, which is wrong. Once the group variable is modified to
include the \texttt{drawID}, the result is a fitted line for each
sample, and also the choice to include a standard error ribbon faintly
in the background for each sample.

\begin{figure}

\centering{

\includegraphics[width=0.8\textwidth,height=\textheight]{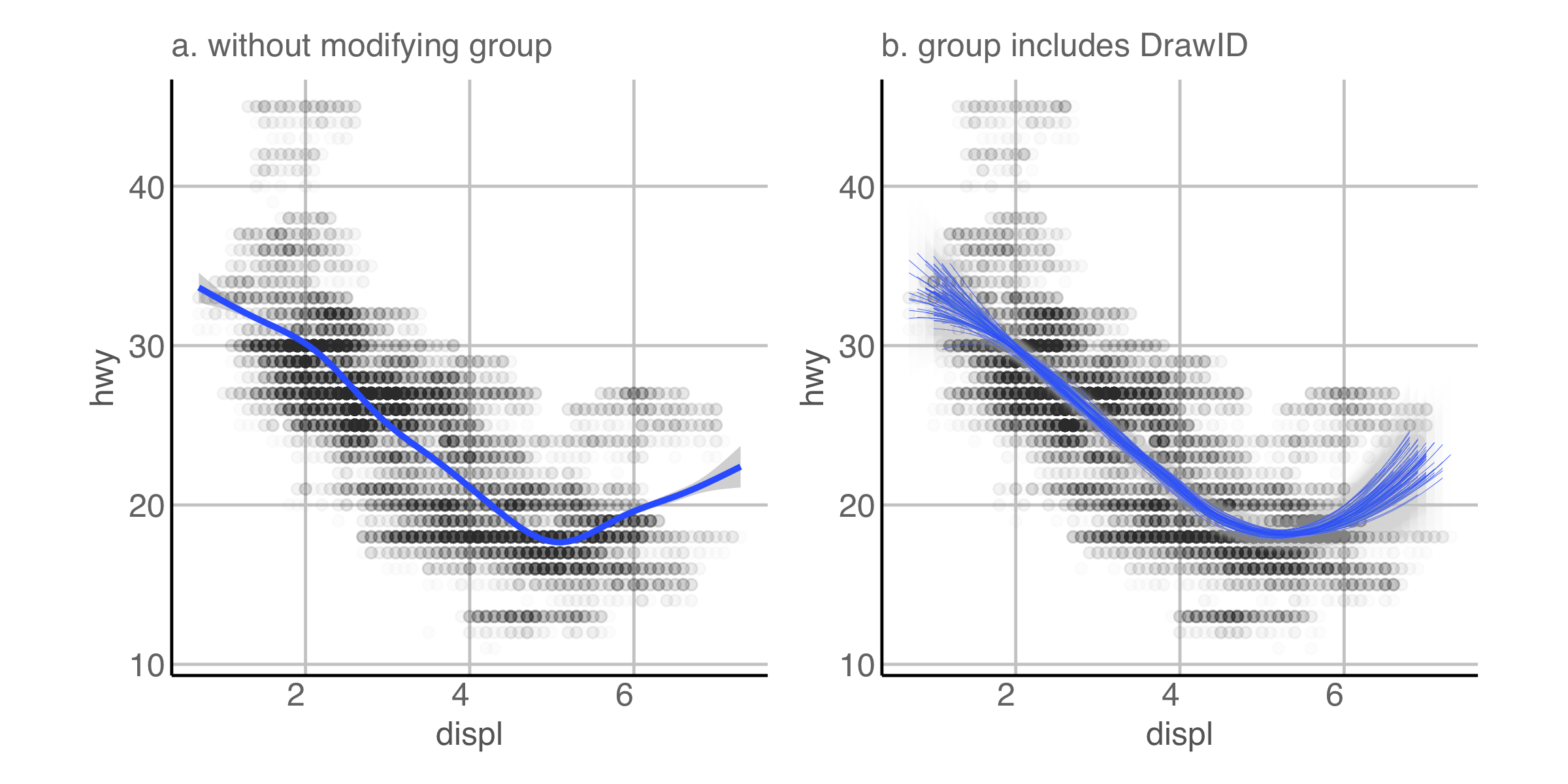}

}

\caption{\label{fig-grouping-need}Modifying the group variable is
essential for handling samples: (a) not done, giving an incorrect
representation of the uncertainty, (b) group variable includes the
\texttt{drawID}. We can see that we need to pass our samples through the
visual function in batches to ensure that the statistics are not
artificially changed by the sample size.}

\end{figure}%

The requirement for samples and \emph{only} samples as our distribution
representation is why the formalisation by \citet{Kay2023}, despite
having the insight to use distributional inputs, did not have the full
flexibility required for EDA. By allowing flexible distribution
representations, \texttt{ggdist} is focused on looking at distribution
as values in their own right, rather than integrating uncertainty into
existing visualisation systems. This is also how the visual functions of
\texttt{ggdist} and \texttt{ggdibbler} in Figure~\ref{fig-dist-example}
diverge from one another. It is important to understand that neither
approach is a subset of the other; they are orthogonal, and most of the
plots made in \texttt{ggdist} cannot be made using \texttt{ggdibbler}.
While there are instances where \texttt{ggdist} and \texttt{ggdibbler}
produce similar-looking plots, these plots cannot be made using the same
data or the same code. The distinction between the two approaches
translates directly from the philosophy of \citet{Mason2026}, who
pointed out that the difference between the role of signal and noise is
in our inferential statistics. By having the distribution statistic
subsume the statistic of the plot, we are changing our inferential
statistic and visualising uncertainty as a signal. This is why we
repeatedly say that \texttt{ggdist} is for looking at uncertainty as
signal, and \texttt{ggdibbler} is for looking at uncertainty as noise.
The visualisations of \texttt{ggdist} cannot be made by
\texttt{ggdibbler} due to the limitations in the distribution statistic,
and the visualisations of \texttt{ggdibbler} cannot be made by
\texttt{ggdist} due to the limitations in the \texttt{ggplot2} level
statistic. Ultimately, both packages are required if we want a complete
picture of uncertain data.

\subsection{\texorpdfstring{The adjustment to
\texttt{geometry}}{The adjustment to geometry}}\label{the-adjustment-to-geometry}

The geometry component of the grammar translates our data to a magnitude
in space \citep{Leland2005}. By displaying each distribution as a
sample, we have distilled uncertainty visualisation down to a simple
over-plotting problem; where we previously had a matrix, \(X_{n,m}\), we
now have an array, \(X_{n,m,t}\). Therefore, to pass our data through
the following stages of the grammar, we need to flatten our array back
into a matrix in such a way that we ensure each outcome from our random
matrix is equally weighted. Over-plotting is usually managed by the
geometry component of our visual function by using position adjustments
such as dodging to prevent overlap or transparency to make the
overlapping visible \citep{Cook2016, ggplot2, Leland2005}. Position
adjustments are an umbrella term used to describe any small changes to
the position of a geometric object, and are usually implemented to
ensure different groups are equally visible. Unlike statistics, we can
perform position adjustments one after the other, which means position
adjustment can be nested within each other. Therefore, we can make the
position adjustment required for the overplotting created by the samples
\emph{within} any position adjustment that already existed in the plot.
This is the final change we will make to our visual function to allow
the visualisation of random variables, and it is formalised in
Definition~\ref{def-geom}.

\begin{definition}[Generalised
geometry]\protect\hypertarget{def-geom}{}\label{def-geom}

Let \(G^*\): \(R\xrightarrow{} B\) be a geometry function that
transforms an array of data into a matrix of geometric positions. We can
further decompose \(G^*\) into \(G = E\circ P^*\), where
\(E:R\xrightarrow{} B\) is the geometry function, and
\(P^*:B\xrightarrow{} B\) is the position adjustment. Let \(C_{n,m,t}\),
\(B_{n,m,t}\), \(A_{n \times t,m}\) be matrices of geometric positions.
We can further decompose \(P^*\) into
\(P^* = P_{within} \circ P_{between}\), where

\begin{itemize}
\tightlist
\item
  \(P_{within}\) is a within sample position adjustment that transforms
  \(C_{n,m,t}\) to \(B_{n,m,t}\), by identifying overlapping points on
  \(C_{n,m,i}\) \(\forall i = 1,...,t\) and applying the specified
  sample position function, and.\\
\item
  \(P_{between}\) is a between sample position adjustment that
  transforms \(B_{n,m,t}\) to \(A_{n \times t,m}\) by flattening
  \(B_{n,m,t}\) into \(B_{n\times t,m}\), identifying overlapping points
  on \(B_{n\times t,m}\) and applying the specified position adjustment.
\end{itemize}

\end{definition}

We will spend the rest of this section detailing the nested position
system.

\subsection{Nested position
adjustments}\label{nested-position-adjustments}

Position adjustments change the location or size of a geometric object
to prevent overlap and ensure all geometric objects remain visible. They
achieve this by placing objects beside one another (dodging), stacking
them on top of each other (stacking), placing them in front of each
other and making them see through (transparency), or showing them one
after each other in quick succession (animations). These options make up
the four dimensions that we have available for position adjustments: x
(dodge), y (stack), z (transparency), and time (animation). Including
transparency and time as position adjustments is not the standard
approach in the literature, as these plots are typically considered
separate axes. This is not unique to uncertainty visualisation, even
\citet{Leland2005} did not discuss positions relative to the axis of
x-y-z-t, but rather specified position adjustments as being on the
measured scale (stack) or in the spare space (dodge). This is an
important distinction, but we choose to frame the positions in terms of
x-y-z-t to highlight that there may be multiple measured or spare axes
in a single plot.

Unlike the statistics, (most) position adjustments do not meaningfully
change the inferential statistic of our plot, so we can nest them freely
without concern. This is particularly useful because without nested
position adjustments, we would need to apply the same position
adjustments to the over-plotting caused by both the original grouping
and the sampling. Since we cannot use stacking on the measured axis for
our samples, this would prohibit us from making uncertain versions of
stacked bar charts, which, again, would be an undesirable limitation to
our uncertainty visualisation system. We can see this problem in
Figure~\ref{fig-positions}, which shows the visualisation of a stacked
bar chart of the \texttt{mpg} data (a), alongside a visualisation of its
random counterpart, the \texttt{uncertain\_mpg} data, visualised using a
``stack'' (b), ``stack\_dodge'' (c), and ``stack\_identity'' (d),
position adjustment. The fact that stacking is not a viable approach
should be obvious: the scale has been artificially inflated, and the
visualisation provides little to no information about our data. In the
case of a bar chart, stacking is aligned with our measurement axis (y),
which leads to it being a problematic adjustment, as our between-sample
position adjustment, \(P_{between}\), can only be implemented on an axis
representing spare space. In other words, stacking is only a viable
position adjustment when the sum of the stacked groups holds meaning
\citep{Leland2005}, which is not true for an arbitrary number of
samples. Interestingly, this split of appropriate versus inappropriate
position adjustments is opposite to the findings of
\citet{Bartonicek2025}, who found that stacking was the only appropriate
axis to use for interactivity, for similar arbitrary scale change
issues. This suggests the possibility of an underlying orthogonal
relationship between uncertainty and interactivity in statistical
graphics that would allow us to implement both interactivity and
uncertainty visualisation simultaneously.

\begin{figure}[t]

\centering{

\includegraphics[width=1\textwidth,height=\textheight]{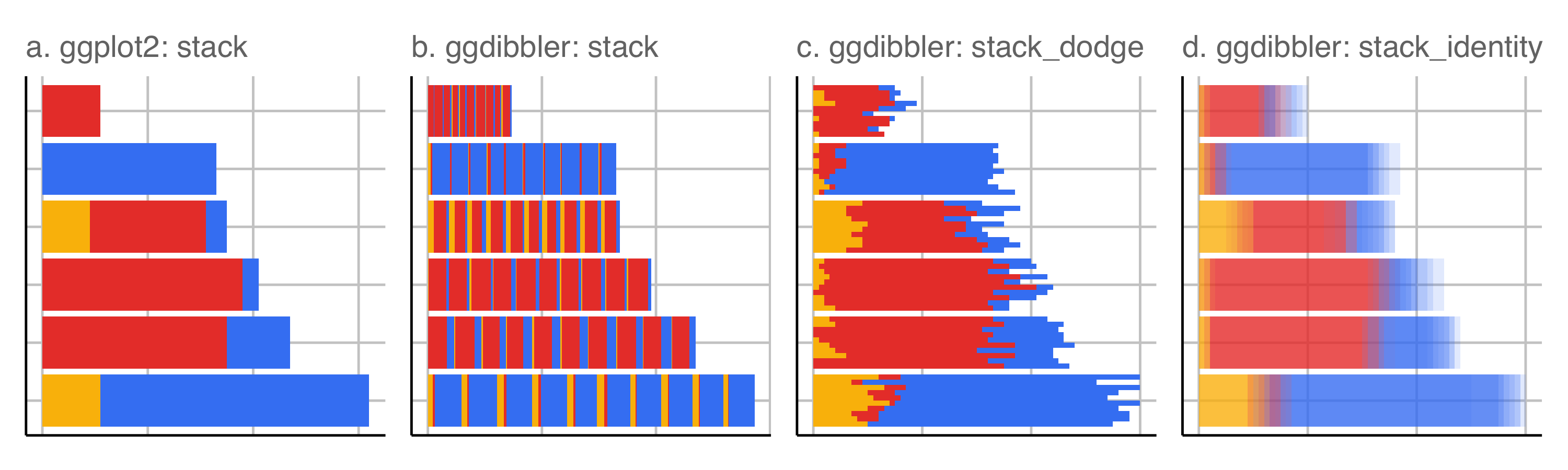}

}

\caption{\label{fig-positions}Why we need nested position adjustments
illustrated using stacked bar charts made using different position
adjustments. Plot (a) shows what a deterministic plot looks like for
reference, while plots (b), (c), and (d) use the same visual function,
but have a random variable input. We can see that stacking is not viable
as plot (b) is unreadable and does not maintain continuity, while
dodging (c) and transparency (d) work well. It is clear that we should
not use the measurement axis for our samples' position adjustment.}

\end{figure}%

Using this framework, we find that a lot of plots with distinct names
can actually be described by a single plot with different position
adjustments. For example, if we have a map with the fill of each area
represented by a random variable, then we could capture this uncertainty
using a HOPs \citep{Hullman2015}, a pixel map \citep{Vizumap}, or a
value-suppressing uncertainty palette \citep{Correll2018}. These plots
could all be made in \texttt{ggdibbler} using a
\texttt{geom\_sf\_sample} and an animation, subdivide (a simultaneous
dodge and stack), or transparency position adjustment, respectively. If
we consider a facet to be a ``between'' plot position adjustment, in
contrast to the ``within'' plot position adjustments we get with dodging
and transparency, we can extend this idea further. Under this framework,
the uncertainty visualisations that map a null distribution with an
alternative on the same plot
\citep{Guo2024, Hullman2021, Savvides2019, McNutt2020} are just the
lineup protocol \citep{Buja2009} without the ``between plot'' position
adjustment.

The most appropriate position adjustment to use can depend on which
aesthetic the random variable is mapped to, and impacts our ability to
read the plot. Figure~\ref{fig-rightposition} illustrates the change in
plot appearance when a random variable is mapped to text using a
transparency (a) and jitter (b), or mapped to colour using a dodge (c)
and transparency (d). We can see that transparency works quite well for
text, while position adjustments such as jitter make the overlapping
text harder to read, regardless of the uncertainty in the estimate. We
can see that colour works well with dodged positions, as it allows us to
see the full sample and do the final calculation visually. Managing
colour with transparency will still produce a technically correct plot,
but it can lead to colours that cannot be matched to the legend, as high
variance colours mix and create new colours that do not belong to the
palette. Differences in the most appropriate position adjustment can
cause conflict when there are multiple sources of uncertainty in a plot.
It would be interesting to investigate this further with a perceptual
experiment to test the effectiveness of different position adjustments
for different aesthetics, but that is well beyond the scope of this
paper.

\begin{figure}[t]

\centering{

\includegraphics[width=1\textwidth,height=\textheight]{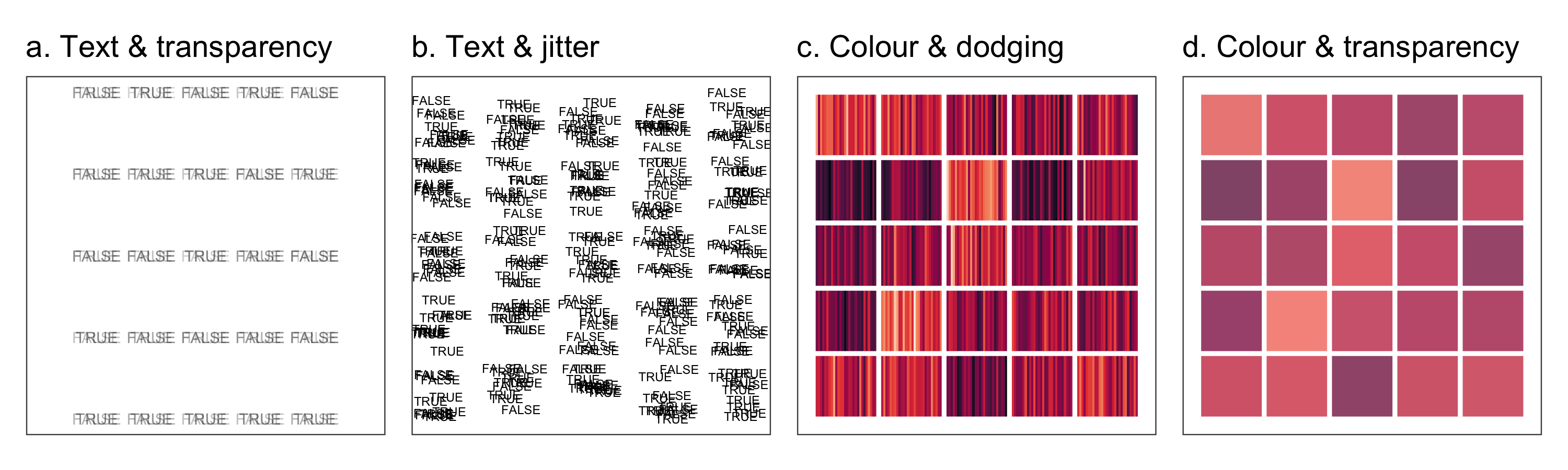}

}

\caption{\label{fig-rightposition}Illustration of the change in plot
appearance based on aesthetic mapping and position adjustment. Plots (a,
b) map the random variable to text with transparency and jitter,
respectively, and plots (c, d) map to tile colour using dodging and
transparency. Although this needs experimental evidence, mapping the
samples to transparency improves readability for text, but for colour,
dodging produces better readability than transparency.}

\end{figure}%

\subsection{The generalised visual
function}\label{the-generalised-visual-function}

If we combine the definitions from Definition~\ref{def-vfunc},
Definition~\ref{def-scale}, Definition~\ref{def-stat}, and
Definition~\ref{def-geom} we have a generalised visual function that
accepts random variable inputs. In this generalised visual function, the
deterministic graphics made by Definition~\ref{def-vfunc} are simply a
special case of Definition~\ref{def-vgeneral}, where every cell is a
degenerate distribution.

\begin{definition}[The generalised visual
function]\protect\hypertarget{def-vgeneral}{}\label{def-vgeneral}

Let \(\textbf{X}\) be a random matrix on the probability space
\((\Omega, \mathcal{F}, Pr)\). Let \(V\) be a function that maps
\(\textbf{X}\) from \((\Omega, \mathcal{F}, Pr)\) to the space of all
visual statistics, \(\Psi\). The visual function, \(V\), can be
decomposed into the following composite function:

\[
V = A \circ O \circ G^* \circ S^* \circ M
\]

\end{definition}

\section{\texorpdfstring{Implementation in
\texttt{ggdibbler}}{Implementation in ggdibbler}}\label{implementation-in-ggdibbler}

The visual function given by Definition~\ref{def-vgeneral} should allow
you to make an uncertainty visualisation version of any graphic that can
be described with the grammar of graphics. We have implemented this
theory in the R package, \texttt{ggdibbler}, which is a \texttt{ggplot2}
extension that allows users to create an uncertain version of any
\texttt{ggplot2} graphic. Figure~\ref{fig-illustration} illustrates this
flexibility by showing a collection of plots that were all made using
this generalised visual function. We can see that there is flexibility
in both plot type, as we display line plots, maps, pie charts,
histograms, bubble charts, and network diagrams, and in aesthetics, as
position, colour, size, slope, and other aesthetics all have a random
variable mapped to them. A single plot can even have multiple sources of
uncertainty simultaneously mapped to different aesthetics. By
establishing a set of rules that will almost always work, we save
ourselves from having to design bespoke software for every single
individual case.

\begin{figure}[t]

\centering{

\includegraphics[width=0.9\textwidth,height=\textheight]{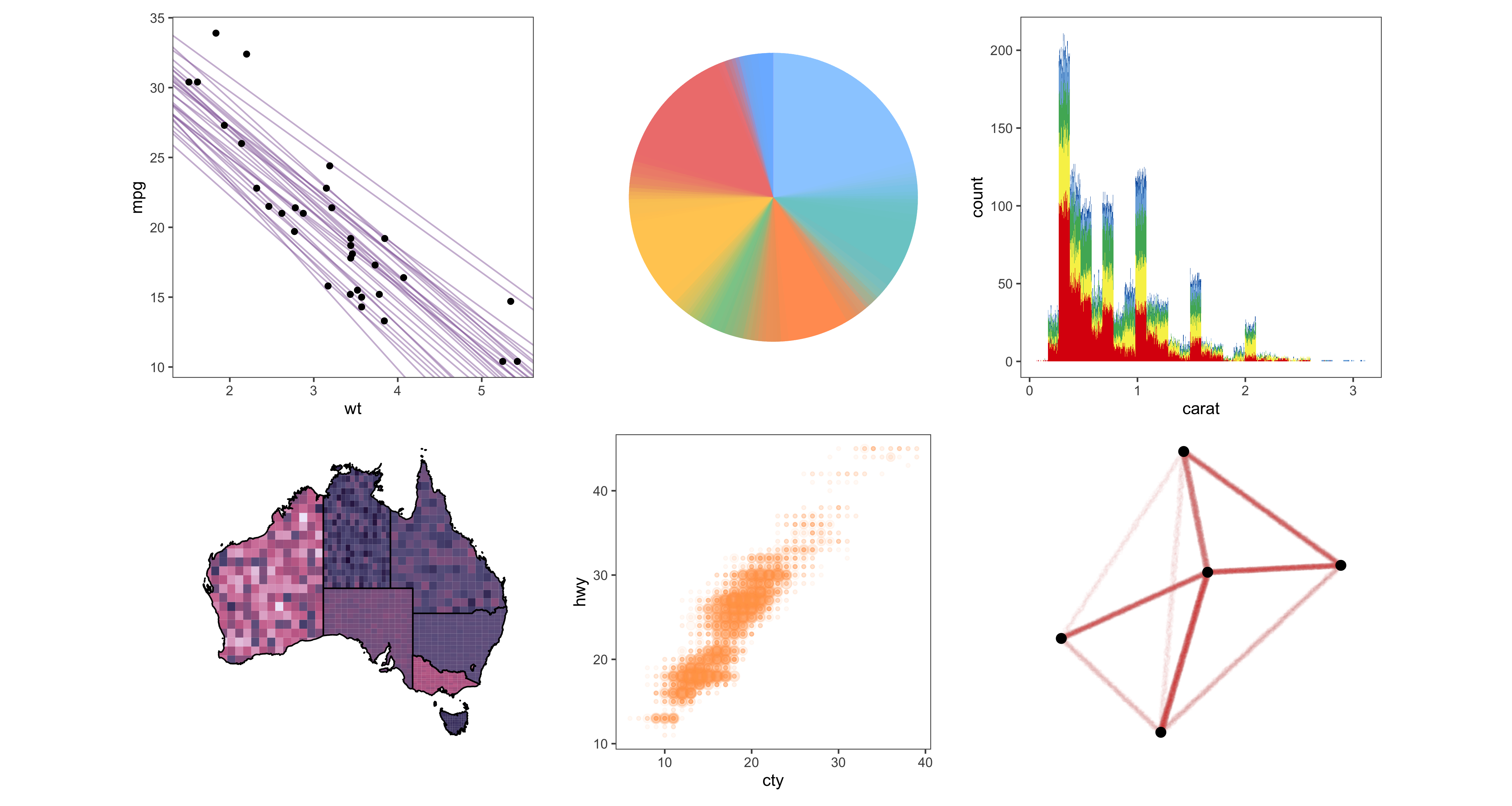}

}

\caption{\label{fig-illustration}This formalisation of uncertainty
visualisation offers extensive flexibility, illustrated by six plots:
line, pie chart, histogram, map, bubble chart, and network diagram.
These plots are made with almost identical syntax with
\texttt{ggdibbler} as that of the deterministic \texttt{ggplot2}
equivalent. These aesthetics - position, colour, size, slope - are all
mapped using random variables.}

\end{figure}%

While we have established the conceptual theory that would underpin a
flexible uncertainty visualisation system, there are considerations that
need to be made when implementing the theory as practical software.
Specifically, we should discuss the design of the user interface, the
data objects that allow us to work with random matrices, and the
computational complexity that comes with uncertain plots. This section
will detail these components.

\subsection{The software design}\label{the-software-design}

The way that a user interacts with a piece of software can help
communicate the mechanisms or theory that underpins it. The statistical
theory behind random matrices, visual convergence, and the continuous
mapping theorem that underpins the \texttt{ggdibbler} package might be
too complicated for someone trying to make a simple scatter plot, but a
basic understanding of these ideas is required to use the package
correctly. Therefore, when designing the functions, we opted to subtly
communicate these ideas through coding paradigms and function names.

Readers familiar with programming paradigms might look at
Definition~\ref{def-vfunc} and Theorem~\ref{thm-cmt} and immediately
think of object-oriented programming (OOP). These readers would be
right, \texttt{ggplot2} can be considered to be an object-oriented
system with a functional-feeling interface. Actually, the bulk of
functionality underlying R structures can be considered to be
object-oriented.

In OOP systems, data is stored as objects and methods that operate on
these objects. The user interacts with these objects only through the
methods, not by directly inspecting the elements. Objects can inherit,
so special objects have features that will work in a variety of
settings, and also some additional special features. Different objects
can respond to the same method call in different ways (polymorphisms).

Today's R contains several choices in data management: S3, S4, R6, and
the latest, S7. S3 forms the original framework, and an example of the
polymorphism is the \texttt{print()} function, where what is printed
will change depending on the object provided. For example, a
\texttt{data.frame} will be printed differently from an \texttt{lm}
(linear model object). It lacks the full characteristics of OOP, though,
because there are no formal class definitions, and it is easy to misuse.
S4 is stricter and underlies all of the Bioconductor
\citep{bioconductor}, but more cumbersome for the user. S7 is designed
to have the ease-of-use of S3 but the strict checks of S4. Of course,
this is not an exhaustive list, and there are several other choices
floating around for specific use-cases, including ggproto and R6.

The main reason to use OOP is polymorphism, which allows us to apply the
same function to different classes of input \citep{Wickham2019}. For
example, this paradigm would allow us to create a \texttt{draw} function
that accepts both \texttt{polygons} and \texttt{lines}, drawing either
one without any special input from the user. The \texttt{sf} package's
\texttt{geom\_sf} function allows users to indiscriminately pass points,
lines, and polygons to the function, which responds exactly as the user
expects, plotting geoms appropriate to the spatial class passed in. OOP
is the logical approach for uncertainty visualisation, as we want to
convey that \(V\) is the same function regardless of the input type.
From the perspective of the user, the function should be the same
whether we input \(\textbf{X}\) or \(X\).

Not only is OOP theoretically consistent with our goals, but it is also
the most practical approach. Using OOP allows us to hide unnecessary
complexity, which is particularly helpful for uncertainty visualisation,
as the abstract nature of uncertainty causes it to be quite error-prone.
Visualising uncertainty without the guardrails put up by
\texttt{ggdibbler} can lead to cases where uncertainty is visualised as
a separate variable, and the graphic does not have the statistical
properties guaranteed by the package. This approach will also make it
easier to perform EDA, as the most flexible programming systems make no
assumptions about our data; instead, they react to the object input by
the user to determine the initial mapping \citep{Leland2005}. This
principle is already carried through in \texttt{ggplot2}, where plots
automatically adapt to categorical, continuous, discrete, or date-time
inputs by having an OOP scale, \(M\). While it would be nice to
implement \texttt{ggdibbler} as a scale, an uncertainty visualisation
system requires more intensive changes to \(V\).

The implicit relationship between \texttt{ggdibbler} and
\texttt{ggplot2} is communicated through the syntax of the code.
Ideally, if all of \texttt{ggplot2} were built on an OOP system, we
could just create an ``uncertainty'' version of all the \texttt{ggplot2}
functions, allowing users to pass distributions without even noticing
the change in the underlying package. That is, both the \texttt{ggplot2}
and \texttt{ggdibbler} plots from Figure~\ref{fig-dist-example} would
use the syntax,
\texttt{ggplot(data\ =\ density\_data)\ +\ geom\_density(aes(x\ =\ x))}
where the visualisation software runs the \texttt{ggdibler} code if
\texttt{x} is a \texttt{distributional} object. As \texttt{ggplot2} is
not built on an OOP system, this is not possible, so instead,
\texttt{ggdibbler} adds a \texttt{*\_sample} suffix in the function
name. This allows us to maintain similar naming conventions to the
related \texttt{ggplot2} function, while also being explicit about what
the function does. For example, the code that makes the \texttt{ggplot2}
density is
\texttt{ggplot(data\ =\ density\_data)\ +\ geom\_density(aes(x\ =\ xmean))},
while the code that makes the \texttt{ggdibbler} density is
\texttt{ggplot(data\ =\ density\_data)\ +\ geom\_density\_sample(aes(x=xdist))}.
This syntax still conveys the idea that the visual function is
identical; it is only the input that has changed.

This strong theoretical foundation not only gives us an intuitive
function design, but it also allows us to have a lot of versatility
built on a shockingly simple code base. Despite the package covering the
full range of geoms in \texttt{ggplot2}, the bulk of \texttt{ggdibbler}
is a single function that does the sample and group adjustment, and a
second function that nests the positions (if needed). Every
\texttt{geom\_*} has a \texttt{geom\_*\_sample} counterpart that calls
these core functions, and then passes the data through the standard
\texttt{geom\_*} pipeline, with little to no bespoke adjustments for
individual geoms. A significant amount of the simplicity of the design
comes from the packages dependencies: \texttt{distributional}
\citep{distributional}, \texttt{ggplot2} \citep{ggplot2}, \texttt{dplyr}
\citep{dplyr}, \texttt{rlang} \citep{rlang}, \texttt{lifecycle}
\citep{lifecycle}, \texttt{scales} \citep{scales}, \texttt{tidyr}
\citep{tidyr}, \texttt{tibble} \citep{tibble}, \texttt{cli} \citep{cli},
and \texttt{sf} \citep{sfpack}.

\subsection{\texorpdfstring{Representing uncertainty using
\texttt{distributional}}{Representing uncertainty using distributional}}\label{representing-uncertainty-using-distributional}

The theoretical framework we have presented assumes you already have a
random matrix input, and thus far, we have somewhat ignored how you
would go about making one. This is because quantifying uncertainty and
representing it as a data object is fundamentally part of the data
manipulation stage of our analysis, so it should be kept as separate as
possible from the visualisation stage \citep{ggplot2}. This gives users
full transparency in \emph{what} precisely they are visualising
\citep{ggplot2}. In \texttt{ggdibbler}, this separation is created by
leveraging the \texttt{distributional} package \citep{distributional},
which allows users to store distributions inside data frames as
individual cells. \texttt{ggdibbler} works from the assumption that
users have already quantified the uncertainty as distributions before
attempting to visualise it, building a system that allows for
\texttt{distributional} inputs.

The \texttt{ggdibbler} package is designed to accept any
\texttt{distributional} input. There are two ways you can define a
distribution in \texttt{distribtional}:

\begin{enumerate}
\def\labelenumi{\arabic{enumi}.}
\tightlist
\item
  A theoretical distribution defined by the distribution and its
  parameters, or
\item
  An empirical distribution defined by a set of samples
\end{enumerate}

Most types of uncertainty can be represented as one of the two cases, so
the software is surprisingly flexible. The purpose of the distribution
is to accommodate classical statistics, or even Bayesian thinking, where
the distribution is known or specified. The reason for being able to
specify the distribution through samples is to accommodate the common
situation today that the distribution may not be known, but we can
describe it with samples. Such a situation might arise when we have
bootstrap samples describing the variability.

The \texttt{ggdibbler} software implements a full uncertain replication
of \texttt{ggplot2}, including the examples. This requires
\texttt{ggdibbler} to have uncertain versions of all the data sets used
in the\texttt{ggplot2} examples, made in distributional, including (but
not limited to) \texttt{uncertain\_mpg}, \texttt{uncertain\_faithful},
and \texttt{uncertain\_economics}. These examples make use of both the
theoretical and empirical distribution types. This allows users to
immediately start using \texttt{ggdibbler} with familiar data and plots,
making the system less daunting to integrate into an existing data
analysis pipeline.

\subsection{Additional computational
complexity}\label{additional-computational-complexity}

Replacing a single plot with a large sample of plots can significantly
increase the computational complexity of the visualisation. The more
random variables we feed into the plot, the bigger the sample size needs
to be, but the bigger the sample size gets, the more computationally
expensive the plots become to make and render. These are simply
manifestations of the classic statistics trade-off between computational
cost and accuracy, and the curse of dimensionality.

In \texttt{ggdibbler}, this sample size is controlled by the
\texttt{times} argument, which is set to \texttt{10} by default.
Sometimes this is appropriate, sometimes it is not; it depends on the
variance of the distributions, the particular plot type, the number of
random variables fed in, etc. We cannot set a reliable and sensible
\texttt{times} argument that works for every plot in every situation, so
instead we advise you to pick a sample size that allows your plot to
visually converge. This is a different type of convergence from the
convergence to a deterministic \texttt{ggplot2} discussed earlier.
Technically, a \texttt{ggdibbler} visualisation is a random variable, so
every time you print one, it will draw a new random sample and look
slightly different from the previous renderings. If your sample size is
big enough, the variability between each visualisation should be small
enough that your conclusions do not change between renderings. We could
take this a little further and suggest that your sample size should be
large enough that different renderings are visually indistinguishable
from one another. In other words, they have visually converged by
Definition~\ref{def-converge}. To be more specific, let
\(V(\textbf{X})_i\) be the \(i^{th}\) rendering of a \texttt{ggdibbler}
plot, then we would say that the plot has converged (and our sample size
is large enough) when
\(d(V(\textbf{X})_i, V(\textbf{X})_j)=0 \forall i,j\). This process of
visual convergence is shown in Figure~\ref{fig-correct-times}. We can
see that the full shape of the distribution becomes more visible as the
\texttt{times} argument increases (the input is a scatter plot of
uniform distributions).

\begin{figure}[t]

\centering{

\includegraphics[width=0.9\textwidth,height=\textheight]{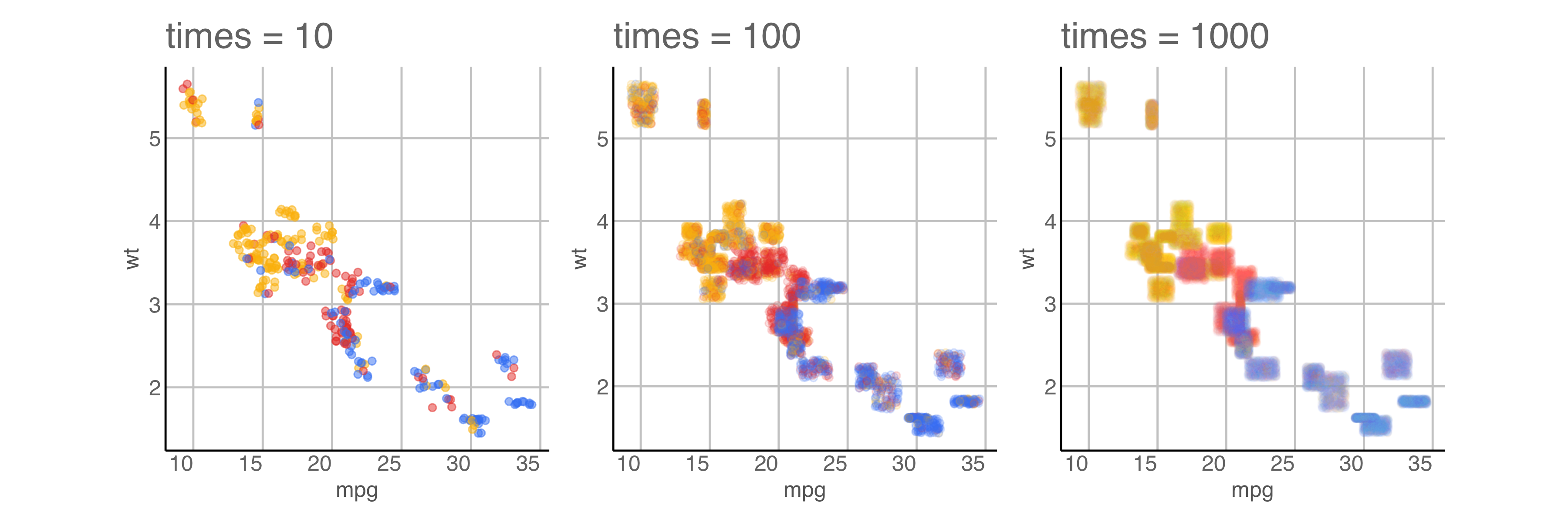}

}

\caption{\label{fig-correct-times}A scatter plot of a random matrix
version of the \texttt{mtcars} data from the `datasets' R package, with
the aesthetic mapping \texttt{x=mpg}, \texttt{y=wt} and
\texttt{colour=cyl}. All three variables are random variables. This plot
shows the impact of an appropriately chosen times argument. We can see
that as the sample size increases, the distributions form cohesive units
and stop looking like a collection of separate points with little
connection. This is not always achievable due to computational cost, but
we should, at the very least, select a sample size that means our
conclusions are not changing between renderings of the plot.}

\end{figure}%

Given this additional complexity, we might opt to skip the entire
sampling procedure and instead just map the ``uncertainty'' to some kind
of aesthetic. After all, despite only using samples to depict
uncertainty, the \texttt{ggdibbler} documentation (along with this
paper) is littered with plots that appear to be using the explicit
mappings that are commonly associated with uncertainty. Blur, fuzziness,
colour lightness, and shape/size frequently appear despite never being
controlled through the aesthetics. Therefore, it might be more
computationally effective to just control these elements directly.

While this idea is not unreasonable and might be possible to implement
in the future, we do not take this approach in \texttt{ggdibbler}, as
directly mapping uncertainty seems to be fraught with danger. The source
of the uncertainty and its visual expression are tightly linked in a
\texttt{ggdibbler} plot. Take, for example, the appearance of blurring
and fuzziness shows in Figure~\ref{fig-fuzzy-blur}, which depicts
\texttt{ggdibbler} fuzziness in a bubble chart, \texttt{ggdibbler} blur
in a dotplot, and \texttt{ggdist} blur in a dotplot. Both blur and
fuzziness use a transparent position adjustment, but blur is created
through uncertainty in position, while fuzziness is created through
uncertainty in size. This allows us to convey multiple types of
uncertainty at once, as blur makes it more difficult to read the
position but does not affect our ability to read the size, and vice
versa for fuzziness. The inability to convey multiple types of
uncertainty is the most common difficulty faced by uncertainty
visualisation approaches
\citep{geointerviews, Hadjimichael2024, MacEachren2005}. The blur and
fuzziness in the \texttt{ggdibbler} plots is more of an ``emergent''
aesthetic, which appears through the combination of transparency and
randomness; this is in direct contrast to the \texttt{ggdist} blur,
which is controlled top-down. This type of plot is rather unusual for
the \texttt{ggdist} package, as the uncertainty is controlled by a
standard error, rather than a distributional input, and that error is
mapped to a ``blur'' aesthetic. Looking closely at the plot, we can see
a blurred cliff effect, where the blurring of some dots extends over the
dots beneath them. Since the dots in a dotplot must be stacked on top of
each other, this type of blurring would not be possible. This makes it
unclear as to how the blur should be interpreted, and it indicates some
kind of breakdown in the relationship between the data and its visual
representation.

\begin{figure}

\centering{

\includegraphics[width=1\textwidth,height=\textheight]{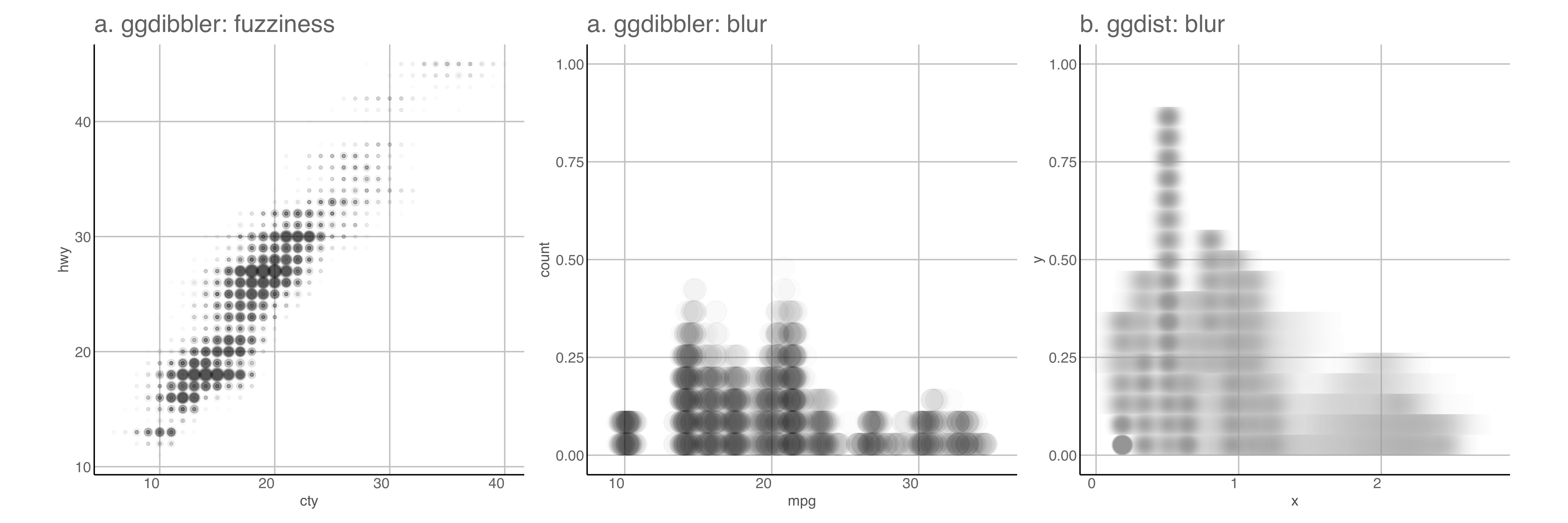}

}

\caption{\label{fig-fuzzy-blur}Two examples from the \texttt{ggdibbler}
documentation, and one example from the \texttt{ggdist} documentation to
illustrate the difference in the top-down versus emergent aesthetic
approach. The blur and fuzziness emerge from the \texttt{ggdibbler}
plots due to the sampling procedures, while the blur in \texttt{ggdist}
is added manually as a top-down aesthetic. We can see that the `cliff'
effect in the \texttt{ggdist} plot is not visible in the blurred
\texttt{ggdibbler} plot, because it would be impossible to generate that
appearance from the underlying data.}

\end{figure}%

Breaking the connection between the data and its visual representation
would result in us losing the desirable statistical properties that are
guaranteed by \texttt{ggdibbler}. This breakdown appears to be quite
common when we try to manually map uncertainty to an aesthetic
\citep{Mason2026}. Additionally, a flexible EDA system should allow
\emph{any} combination of \emph{any} uncertain aesthetics, and working
out how these combinations of random variables should appear would be
incredibly laborious, if not impossible. Additionally, this approach
would almost defeat the purpose of the system, as the whole point of
visualising the data to begin with is that we \emph{don't already know}
what it should look like. For these reasons, trying to directly map
uncertainty to an aesthetic might be better computationally, but would
likely result in more headaches than it would be worth, if it were
possible at all.

\section{Conclusions and future
research}\label{conclusions-and-future-research}

This paper set out to formalise uncertainty visualisation in order to
design a flexible uncertainty visualisation system that maintains nice
statistical properties. The value of this formalisation is not only in
the power of a truly flexible system, but also in the mathematical
foundation itself that ensures the connection between data and visual
aesthetic is always maintained.

By defining the visual function mathematically, we leveraged the concept
of continuity to evaluate the behaviour of visual statistics. This
approach uncovers a wealth of other statistical concepts we could
translate to statistical graphics. Building on this work, we should
investigate other concepts like bias/variance trade-off, statistical
sufficiency, and other convergence properties for visual statistics. In
this vein, we should also identify the minimum level of variance in a
plot that is required for the human perception of visual convergence,
related to the concept of ``just noticeable differences'' in
psychophysics. Similarly, it would be useful to know if there are
principles for determining the sample size required for visual
convergence, as it would give us a convenient ceiling on the
computational complexity of the plots.

On the topic of convergent graphics, it would be interesting to find out
if all plots with the same limiting visual statistic have some
underlying computation in common. This curiosity is quite similar to the
question posed by \citet{ggplot2}, where they asked if plots that are
visually identical, that are made using different grammar adjustments,
have some underlying principle in common. When looking into interactive
visualisations \citet{Bartonicek2025} also found underlying links
between the mathematical functions we plot and the appropriate visual
representations, so we wonder if there is a more ``fundamental'' version
of the grammar of graphics that can make this link more explicit. In
mathematics, all functions can essentially be broken down into a basic
increment/step function (something akin to stacking in graphics), so we
do wonder if a similar principle can be applied to plots.

The motivation for a more fundamental building block of statistical
graphics can also be found in our brief discussion on emergent
aesthetics. The emergent aesthetics blur the line between the statistic
and geometric stages of plot building, begging the question ``Should we
explicitly map a statistic, or allow it to emerge through the
visualisation process?'' To solve this problem, we would need to start
by untangling which aesthetics are primary aesthetics and which are
emergent aesthetics. We would also need to understand the conditions
under which each of these aesthetics arises and the position adjustments
that lead to them.

Position adjustments have historically been an afterthought when
building graphics, but this work motivates a more thorough investigation
into how they affect the readability of a plot. We discussed how
text/shape works well with alpha, and dodging/subdividing is good for
colour, but this is just an anecdotal observation. A more formal theory
on how position adjustments change our ability to read a plot is much
needed.

Finally, the \texttt{ggdibbler} software has room for improvement.
Currently, the software only accepts individual distributions and, if
multiple distributions are passed, they are assumed to be independent.
The \texttt{distributional} software allows for joint distributions,
which are functionally a random matrix collapsed into a vector.
Expanding \texttt{ggdibbler} to accept these joint distributions is a
natural extension for the software. Additionally, \texttt{ggdibbler}
does not work with \texttt{ggplot2} extensions due to the way the
software is implemented. We plan to extend the package to make this
possible. A full list of the planned changes is available, along with
the package source code in the \texttt{ggdibbler} GitHub.

\section{Acknowledgements}\label{acknowledgements}

The first author of this paper is supported in part by a scholarship
from the the Australian Energy Market Operator. This research was
supported by the Commonwealth through an Australian Government Research
Training Program Scholarship {[}DOI:
https://doi.org/10.82133/C42F-K220{]}. The first author would also like
to thank Mitchell O'Hara-Wild and Cynthia Huang for their comments and
feedback which substantially improved the work, as well as Ze-Yu Zhong
for several interesting examples that ended up being used in this paper.
The R packages used for this work were: \texttt{tidyverse}
\citep{tidyverse}, \texttt{distributional} \citep{distributional},
\texttt{ggdist} \citep{Kay2023}, \texttt{ggdibbler} \citep{ggdibbler},
\texttt{patchwork} \citep{patchwork}, \texttt{khroma} \citep{khroma},
\texttt{tidygraph} \citep{tidygraph}, \texttt{colourspace}
\citep{colorspace}, \texttt{ggraph} \citep{ggraph}, \texttt{ozmaps}
\citep{ozmaps}, \texttt{sf} \citep{sfpack}, and \texttt{ggthemes}
\citep{ggthemes}. The GitHub repository for this paper can be found at
https://github.com/harriet-mason/paper-ggdibbler, which contains the
files required to reproduce this article in full.

  \bibliography{paper.bib}

\end{document}